\documentclass[11pt,letterpaper,reqno]{amsart}

\pdfoutput=1

\usepackage{fullpage}
\usepackage[foot]{amsaddr}
\usepackage{amsmath}
\usepackage{amsthm}
\usepackage{amssymb}
\usepackage{mathrsfs}
\usepackage{url}

\usepackage{todonotes}
\usepackage{mathtools}
\usepackage{textpos}
\usepackage[T1]{fontenc}

\usepackage[breaklinks,bookmarks=true,hypertexnames=false,pagebackref]{hyperref}
\hypersetup{colorlinks=true, citecolor=blue, linkcolor=red, urlcolor=blue}
\usepackage{cleveref}

\usepackage{thm-restate}

\makeatletter
\def\paragraph{\@startsection{paragraph}{4}%
  \z@\z@{-\fontdimen2\font}%
  {\normalfont\bfseries}}
\makeatother

\newtheorem{theorem}{Theorem}

\newtheorem{lemma}[theorem]{Lemma}

\theoremstyle{definition}

\theoremstyle{remark}
\newtheorem{remark}{Remark}

\AddToHook{env/lemma/begin}{\crefalias{theorem}{lemma}}
\AddToHook{env/claim/begin}{\crefalias{theorem}{claim}}
\AddToHook{env/observation/begin}{\crefalias{theorem}{observation}}
\AddToHook{env/proposition/begin}{\crefalias{theorem}{proposition}}
\AddToHook{env/corollary/begin}{\crefalias{theorem}{corollary}}
\AddToHook{env/condition/begin}{\crefalias{theorem}{condition}}
\AddToHook{env/definition/begin}{\crefalias{theorem}{definition}}
\AddToHook{env/question/begin}{\crefalias{theorem}{question}}
\AddToHook{env/example/begin}{\crefalias{theorem}{example}}
\AddToHook{env/conjecture/begin}{\crefalias{theorem}{conjecture}}
\AddToHook{env/remark/begin}{\crefalias{theorem}{remark}}

\crefname{theorem}{Theorem}{Theorems}
\crefname{observation}{Observation}{Observations}
\crefname{claim}{Claim}{Claims}
\crefname{condition}{Condition}{Conditions}
\crefname{algorithm}{Algorithm}{Algorithms}
\crefname{property}{Property}{Properties}
\crefname{example}{Example}{Examples}
\crefname{fact}{Fact}{Facts}
\crefname{lemma}{Lemma}{Lemmas}
\crefname{corollary}{Corollary}{Corollaries}
\crefname{definition}{Definition}{Definitions}
\crefname{remark}{Remark}{Remarks}
\crefname{proposition}{Proposition}{Propositions}
\crefname{section}{Section}{Sections}
\crefname{equation}{equation}{equations}

\newcommand{\vsupp}{\operatorname{supp}}

\title{Partition Rank and Algebraic Circuit Lower Bounds}

\author{Cornelius Brand, Petteri Kaski, Jiaheng Wang}
\address[Cornelius Brand]{University of Regensburg.}
\address[Petteri Kaski]{Aalto University.}
\address[Jiaheng Wang]{University of Regensburg, \and University of Helsinki.}
\thanks{This project is funded by the European Union (ERC, CountHom, 101077083). Views and opinions expressed are however those of the author(s) only and do not necessarily reflect those of the European Union or the European Research Council Executive Agency. 
JW has received financial support from a postdoctoral fellowship funded by the Helsinki Institute of Information Technology (HIIT)}

\begin{document}

\begin{abstract}
    Strassen's theory of bilinear complexity provides a mathematical characterization of the arithmetic complexity of primitives such as matrix multiplication via the rank of tensors. However, the connection to tensor rank is known to break down in higher degrees of multilinearity.
    
    In this work, we highlight an unexplored connection between a generalized notion of tensor rank, which can be defined in Naslund's framework of partition ranks (JCTA~2020),
    and multiplicative complexity. 
    These partition ranks allow us to control the multiplicative complexity, 
    and thus arithmetic complexity, in any constant degree of multilinearity 
    from below, while recovering Strassen's seminal characterization in the bilinear case. 
    This enables novel potential applications of the rank-based approaches to problems in fine-grained algorithms and complexity, such as the hyperclique conjecture of Lincoln-Williams-Vassilevska Williams (SODA~2018).
    Moreover, we exhibit connections to established notions of rank, such as tensor slice rank (in the sense of Tao and Sawin), as well as its symmetric variant. For computing the latter symmetric variant, we point out a simple NP-hardness proof, contrasting the rather involved NP-hardness proof for ordinary, non-symmetric tensor slice rank by Bläser et al. (SODA~2021).
\end{abstract}

\maketitle

\begin{textblock}{5}(9.6, 5.8) \includegraphics[width=90px]{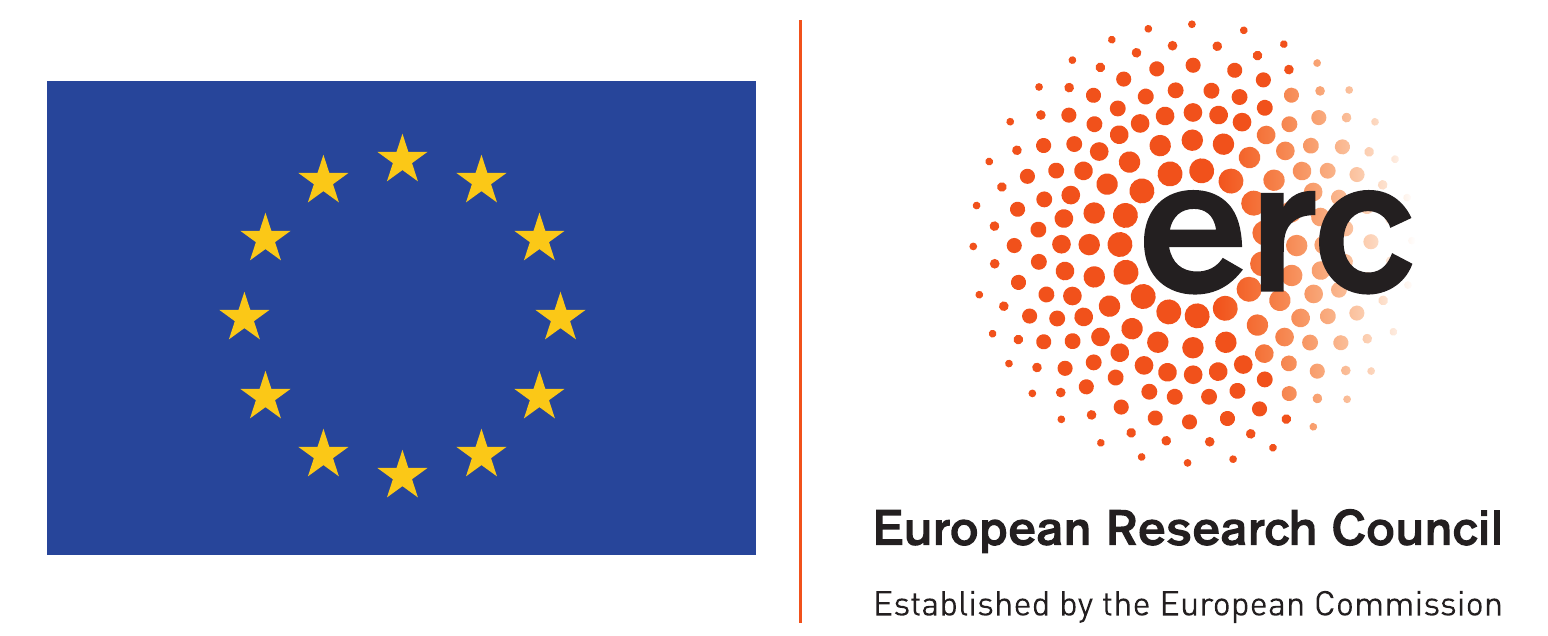} \end{textblock}

\section{Introduction}

Strassen's 1969 breakthrough on algorithms for matrix multiplication~\cite{Strassen1969} has had substantial impact both within and outside of algebraic complexity theory.
Besides this algorithmic result, Strassen also laid out a structural theory around \emph{lower bounds} on the so-called {bilinear complexity} of computation in algebraic structures such as matrices and polynomials. One of the cornerstones of this theory is the result that the number of multiplications between inputs from a field $\mathbb F$ needed to perform such operations using any possible computation scheme (their \emph{multiplicative complexity} $L$) 
is characterized up to a small and explicit factor by the \emph{rank} of an associated tensor~\cite[Korollar 5]{Strassen1973}; see the discussion of~\cite[Eq.~(14.8)]{BurgisserCS1997} for an account in English. For example, Strassen's theory reduces the fast evaluation of a matrix product $W = UV$ for given square matrices $U = (U_{i,j}:i,j\in[n])$ and $V = (V_{i,j}:i,j\in[n])$ into the task of finding low-rank decompositions for the associated $3$-tensor, viewed as a polynomial 
\begin{equation}
\label{eq:mm}
\mathrm{MM}_n = \sum_{i,j,k\in[n]} U_{i,j}V_{j,k}W_{i,k}
\end{equation}
that is linear in each of the three sets of indeterminates, $U_{i,j}$, $V_{i,j}$, and $W_{i,j}$ for $i,j\in[n]:=\{1,\ldots,n\}$.
To state Strassen's main result in the language of polynomials, 
we define the {\em multiplicative complexity} $L(f)$ of $f\in\mathbb F[X_1,\ldots,X_q]$ as the smallest $\ell$ such that $f$ is computed by any computation scheme that performs no more than $\ell$ essential multiplications (formal definitions are given later).

Strassen's main result now connects, up to small constants, the multiplicative complexity of a
bilinear map, viewed as a $3$-tensor or the associated polynomial $f$, with the tensor rank $R(f)$ of the tensor. 
Adapting the constants to the present formalism, this reads:
\begin{theorem}[Strassen~\cite{Strassen1973}] \label{thm:strassen}
    For every\/ $3$-tensor $f \in (\mathbb F^m)^{\otimes 3}$, we have 
    \begin{align}
    \frac{1}{96}\cdot R(f) \leq L(f) \leq 2 \cdot R(f)\,.
    \end{align}
\end{theorem}
The power of Strassen's approach is exemplified by matrix multiplication algorithms: 
the \emph{exponent of matrix multiplication} $\omega$ encodes the minimal (more precisely, infimal) value such that there is an algorithm that computes the product of two $n \times n$ matrices using $O(n^\omega)$ arithmetic operations. 

Essentially all progress on bounding this number from above since 
Strassen's 1969 result has come through increasingly intricate decompositions of the matrix multiplication tensor \eqref{eq:mm} and \cref{thm:strassen}; for many examples hereof, see Bläser's survey~\cite{BlaeserMaMu}.

Strassen's lower-bound theory, however, is known to be strictly bilinear. 
Namely, the connection between arithmetic circuit complexity and tensor rank in \cref{thm:strassen} fails for 
all $d$-tensors ($(d-1)$-linear maps) with $d\geq 4$, where the so-called {\em product of inner products} tensor 
defined for $d=4$ by the polynomial
\begin{equation}
\label{eq:pip}
\mathrm{PIP}_n = \biggl(\sum_{i\in [n]}X_iY_i\biggr)\biggl(\sum_{j\in [n]}Z_jW_j\biggr)
\end{equation}
provides a well-known demonstration; namely the formula \eqref{eq:pip} shows 
$L(\mathrm{PIP}_n)\leq 2n+1$ but we have $R(\mathrm{PIP}_n)=n^2$ by a matrix flattening of the tensor.
Within fine-grained complexity theory, 
Lincoln, Williams, and Vassilevska Williams~\cite[Theorem~8.1]{LincolnWW18} 
highlight this inadequacy of tensor rank in the context of (what is now known as) the hyperclique conjecture.
One way to approach the hyperclique conjecture algebraically is to seek lower bounds for the multiplicative 
complexity of the corresponding polynomials, such as the $3$-uniform $4$-hyperclique polynomial:
\begin{equation}
\label{eq:hyp}
\mathrm{HYP}_n = \sum_{i,j,k,\ell\in [n]}X_{i,j,k}Y_{i,j,\ell}Z_{i,k,\ell}W_{j,k,\ell}\,.
\end{equation}

\subsection*{Our Results}
We show that an alternative tensor parameter to the tensor rank $R(f)$ gives a lower bound for 
the multiplicative complexity, and thus algebraic circuit complexity, of a $d$-tensor
$f\in(\mathbb F^m)^{\otimes d}$ for any $d\geq 3$. This alternative parameter is a symmetric 
instance of Naslund's partition rank~\cite{Naslund2020}. 
Namely, for $\lambda\vdash d$ an integer partition of $d$, let us write $R_\lambda(f)$ 
for Naslund's $\mathcal{P}$-rank of $f$, where $\mathcal{P}$ consists of all set partitions of $[d]$ with size partition $\lambda$. 
We call $R_\lambda(f)$ the $\lambda$-{\em rank} of $f$. 
Two immediate special cases of $\lambda$-rank are 
the tensor rank $R=R_{1,1,\ldots,1}$ and {\em slice} rank $R_{d-1,1}$. 
(We postpone a detailed discussion of
tensor and polynomial parameters to the section on related work below as well as to
the preliminaries in \cref{sect:mult}.)
Our main result in this paper is that the {\em two-slice} rank $R_{d-2,1,1}$ 
controls the multiplicative complexity from below for all $d\geq 3$.

\begin{restatable}[Main; Lower bound for multiplicative complexity of $d$-tensors, $d\geq 3$]{theorem}{mainthm}\label{thm:main}
    For every $d\geq 3$ and every\/ $d$-tensor $f \in (\mathbb F^m)^{\otimes d}$, we have 
    \begin{align} 
    \frac{1}{(d+1)^2d(d-1)}\cdot R_{d-2,1,1}(f) \leq L(f)\,.
    \end{align}
\end{restatable}

We observe that \cref{thm:main} recovers 
Strassen's lower-bound theorem (\cref{thm:strassen}) when $d=3$.
Comparing the slice rank and the two-slice rank, we observe
$R_{d-1,1}\leq R_{d-2,1,1}$ for all $d\geq 3$ (cf.~\cref{rem:refinement}) 
implying that the two-slice rank gives at least as good control for lower bounds 
as the slice rank. 
Beyond these immediate observations, we have two-fold further motivation for \cref{thm:main} and 
the present paper.

Our first motivation for \cref{thm:main} is that the two-slice rank $R_{d-2,1,1}$ gives
{\em better} generic control for lower bounds than the slice rank $R_{d-1,1}$. 
Indeed, for constant $d\geq 3$ we have 
$R_{d-1,1}(f)\leq m$ for all tensors 
$f \in (\mathbb F^m)^{\otimes d}$, whereas heuristically most tensors 
$f \in (\mathbb F^m)^{\otimes d}$ satisfy 
$R_{d-2,1,1}(f)=\Omega(m^2)$ (cf.~\cref{rem:random}).
Thus, in contrast to slice rank, \cref{thm:main} enables superlinear-in-$m$ lower 
bounds on the arithmetic complexity of specific problems encoded by polynomials
$f\in(\mathbb F^m)^{\otimes d}$ via superlinear lower bounds 
on the two-slice rank $R_{d-2,1,1}(f)$. For example, \Cref{thm:main} implies that
one possible strategy to prove the (algebraic) hyperclique conjecture 
$L(\mathrm{HYP}_n)=\Theta(n^4)$ for $d=4$ is to show the lower bound
$R_{2,1,1}(\mathrm{HYP}_n)=\Omega(n^4)$ on the degree-$4$ tensor $\mathrm{HYP}_n \in (\mathbb F^{m})^{\otimes 4}$, where $m=n^3$, in~\eqref{eq:hyp}.

Our second motivation for \cref{thm:main} is an elementary observation
that enables a proof in three short steps, which we now sketch.

\subsection*{\texorpdfstring{Overview of Techniques and a Proof of \cref{thm:main}}{Overview of Techniques and a Proof of Theorem 2}}

Namely, we relate parameters of $d$-tensors to analogous parameters of homogeneous polynomials 
via the basic fact that 
a polynomial $f$ representing a $d$-tensor {\em is} a homogeneous polynomial $f$ of degree $d$,
with the tensors \eqref{eq:mm}, \eqref{eq:pip}, and \eqref{eq:hyp} above giving concrete examples.

The measure of complexity employed for both tensors and polynomials is their multiplicative complexity, which we now formally define.
A sequence $(f_1,\ldots,f_\ell) \in \mathbb F[X_1,\ldots,X_q]^\ell$ of polynomials is a  \emph{multiplication sequence} of \emph{length} $\ell$ if, for all $1 \leq j \leq \ell$, there are $s_j,t_j \in \operatorname{span}_{\mathbb F}\{1,X_1,\ldots,X_q,\allowbreak f_1,\ldots,f_{j-1}\}$ with $f_j = s_jt_j$. We say that the sequence \emph{computes} a polynomial $f$ whenever $f \in \operatorname{span}_{\mathbb F} \{1,X_1,\ldots,X_q,f_1,\ldots,f_\ell\}$. The \emph{multiplicative complexity $L(f)$ of $f$} is the smallest length of a multiplication sequence computing $f$. Note that every ordinary arithmetic circuit (with multiplicative fan-in two) that uses $\ell$ multiplication gates can be converted into a multiplication sequence of length $\ell$, and vice versa. Consequently, a lower bound on $L(f)$ is a lower bound on arithmetic circuit complexity.

Accordingly, to study multiplicative complexity of tensors in their interpretation as homogeneous polynomials, it is convenient 
to work with multiplication sequences $(f_1,\ldots,f_\ell)$ where each $f_i$ is homogeneous of degree $d_i$ for some $d_i\geq 0$ for all $1\leq i\leq \ell$. 
Such multiplication sequences are themselves called \emph{homogeneous}. Whenever $f$ is homogeneous, we write $H(f)$ for the \emph{homogeneous} multiplicative complexity of~$f$, as the smallest $\ell$ such that there is a homogeneous multiplication 
sequence that computes $f$. Every multiplication sequence for a homogeneous polynomial $f$ 
of degree $d$ can be {\em homogenized}\/~\cite[Lemma (21.25)]{BurgisserCS1997}, one multiplication at a time, to yield 
\begin{equation}
\label{eq:homogeneization}
\frac{1}{(d+1)^2}\cdot H(f)\leq L(f) \leq H(f)\,.
\end{equation}
In particular, to control $L(f)$ from below, by~\eqref{eq:homogeneization} 
it suffices to control $H(f)$ from below. 
Such control is afforded by the following lemma, 
which is our main technical contribution.
It connects the homogeneous complexity to the polynomial counterpart of 
the $\lambda$-rank $R_\lambda$, namely the $\lambda$-{\em strength} $S_\lambda$,
whose detailed definition we postpone to \cref{sect:mult}.

\begin{restatable}[Main; Lower bound for homogeneous multiplicative complexity]{lemma}{strlemma}\label{lem:strength}
For every homogeneous polynomial $f$ of degree $d\geq 3$, we have
\begin{equation}
\label{eq:strength}
S_{d-2,1,1}(f) \leq H(f)\,.
\end{equation}
\end{restatable}
Finally, we control the $\lambda$-strength from below by the $\lambda$-rank
by routine set-multilinearization, namely for every $d$-tensor $f \in (\mathbb F^m)^{\otimes d}$ and every $\lambda \vdash d$, we have
\begin{equation}
\label{eq:set-multilinearization}
S_\lambda(f)\leq R_\lambda(f)\leq \frac{d!}{\lambda!}\cdot S_\lambda(f)\,.
\end{equation}
We present a proof for completeness in \cref{lem:set-multilinearization} in \cref{sect:mult}. \cref{thm:main} now follows immediately from
\eqref{eq:homogeneization}, \cref{lem:strength} and \eqref{eq:set-multilinearization}. 

To supply a short intuition on the proof of~\cref{lem:strength}, observe that
every homogeneous multiplication sequence $\Phi=(f_1,f_2,\ldots,f_\ell)$ that computes a 
polynomial of degree at least two must contain $1\leq k\leq \ell$ polynomials $f_{i_1},f_{i_2},\ldots,f_{i_k}$ with degree exactly two. 
Furthermore, since $\Phi$ is homogeneous, each $f_{i_j}=\alpha_j\beta_j$
for $j\in[k]$ is the product of two linear forms $\alpha_j$ and $\beta_j$.
By essentially only the distributive law and a simple induction on $\Phi$
(detailed in the proof of \cref{lem:seqlin} in \cref{sect:mult}), 
these linear forms persist over the linear combinations and 
multiplications in $\Phi$; namely, {\em every} homogeneous polynomial $h$ of 
degree $d\geq 2$ computed by $\Phi$ can be written as $h=\sum_{j\in[k]}\alpha_j\beta_jg_j^{[h]}$ for some homogeneous 
polynomials $g_j^{[h]}$ of degree $d-2$ depending on $h$. 
For $d\geq 3$ this yields $S_{d-2,1,1}\leq H$ and thus~\cref{lem:strength}.

The parameter $S_{d-2,1,1}$, which we call the {\em two-slice} strength,
is thus of independent interest because lower bounds on $S_{d-2,1,1}$ imply immediate and potentially nonlinear 
lower bounds on the multiplicative complexity $L$
by combining \eqref{eq:homogeneization} and \cref{lem:strength}. The two-slice strength
$S_{d-2,1,1}$
is also easier to work with than the two-slice rank $R_{d-2,1,1}$ because all the polynomial 
indeterminates have a symmetric role in monomials.
In contrast, for tensors one has to track which of the $d$ tensor modes are being considered. This is, however, 
not to say that the two-slice strength is an easy parameter computationally; we show that 
computing both the slice strength and the two-slice strength of a given polynomial is NP-hard.

\begin{theorem}[NP-hardness of slice strength and two-slice strength]
\label{thm:np-hardness-of-slice-strength-and-two-slice-strength}
When given as input a list of monomials for a degree $d$ homogeneous multilinear polynomial $f$
with coefficients in $\{0,1\}$ over a field of characteristic zero,
\begin{enumerate}
\item 
it is $\mathrm{NP}$-hard to compute the slice strength $S_{d-1,1}(f)$ for $d\geq 3$, and
\item 
it is $\mathrm{NP}$-hard to compute the two-slice strength $S_{d-2,1,1}(f)$ for $d\geq 4$.
\end{enumerate}
\end{theorem}

\begin{remark}
    For $d=2$, the slice strength $S_{d-1,1}(f)=S_{1,1}(f)$ is reducible to computing the rank of the quadratic form $f$.
    For $d=3$, the two-slice strength $S_{d-2,1,1}(f)=S_{1,1,1}(f)$ is also known as the Chow rank of the cubic $f$. It seems likely that this quantity is hard to compute as well, but as far as we are aware any proof of this would have to pass through highly involved arguments along the lines of Håstad, Shitov and others~\cite{Hastad90,SchaeferS18,shitov2016hardtensorrank,Swernofsky18}. For the sake of simplicity and being self-contained, we do not attempt a proof here.
\end{remark}

Our proof of Theorem~\ref{thm:np-hardness-of-slice-strength-and-two-slice-strength} 
essentially follows the ideas of Sawin and Tao~\cite{Tao2016Capset,TaoSawin2016SliceRank}.
Using a characterization of slice rank of tensors ($R_{d-1,1}$ in our notation) via orthogonal linear subspaces, they formulate an antichain-condition on the support of a tensor which enables an interpretation of determining its slice rank as a combinatorial optimization problem. This correspondence was employed by Bläser et al.~\cite{BlaserILPS21} to prove $\mathrm{NP}$-hardness of computing that quantity for $d=3$. (For $d=2$, the slice rank $R_{d-1,1} = R_{1,1}$ is simply matrix rank, which can be computed in polynomial time.) Our proof of \cref{thm:np-hardness-of-slice-strength-and-two-slice-strength} adapts the original Sawin--Tao antichain-condition to the symmetric case, that is, the slice strength $S_{d-1,1}$, and correspondingly proves $\mathrm{NP}$-hardness of this quantity via minimum hitting sets for antichains in Gale order of degree sequences. We
also extend the ideas to show NP-hardness of the two-slice strength $S_{d-2,1,1}$ for $d\geq 4$.
It should be noted that the proof for slice strength turns out much simpler than the proof for slice rank. Intuitively, this can be attributed to the fact that slice strength symmetrizes over the different choices of which among the $d$ modes of the tensor to select as the ``slices'' per definition, which makes the characterizations simpler to work with and state.

The rest of this paper is organized as follows. 
After an overview of related work, we prove \Cref{lem:strength} in  \cref{sect:mult}, and \Cref{thm:np-hardness-of-slice-strength-and-two-slice-strength} in \cref{sect:hard}. 

\subsection*{Related Work}
Tao~\cite{Tao2016Capset} introduced slice rank as a variant of the standard tensor rank in the context of the cap-set problem, and further investigated it in its own right together with Sawin~\cite{TaoSawin2016SliceRank}. Naslund defined and studied partition ranks~\cite{Naslund2020}, which systematically generalize slice rank and embed it into a structured collection of related notions of rank.
Analogues of tensor rank defined for (homogeneous) polynomials, often referred to as symmetric tensors in this context, have been studied.
This includes the classic Waring and Chow rank~\cite{torrance2017generic}, and the so-called {strength} (or Schmidt rank) of a polynomial.
Strength was introduced under this name in the proof of Stillman's conjecture by Ananyan and Hochster~\cite{AnanyanH20a}.
Its definition can be restricted to obtain slice strength and other variants with prescribed degrees~\cite{bik2022strength}.
These variants of strength for polynomials have been used to obtain superlinear, unconditional lower bounds in the restricted model of homogeneous algebraic branching programs~\cite{GesmundoGIL22}.

In general, the relation between quantities defined on polynomials and their tensor analogues is intricate. Still, some relations between the complexity of tensors (interpreted as set-multilinear polynomials) with respect to measures of complexity defined on polynomials have been established in algebraic complexity theory.
For instance, Raz~\cite{Raz2013} deduced lower bounds on arithmetic formula size from explicit constructions of high-rank tensors (see also \cite{ChillaraKSV16}) via a set-multilinearization argument.
Set-multilinearization also played a role in the recent breakthrough of Limaye, Srinivasan and Tavenas~\cite{LimayeST2024} on unconditional lower bounds against constant-depth arithmetic circuits.
An attempt at transporting some upper bound results from Strassen's bilinear theory to high-degree tensors and circuits was recently made in~\cite{BrandCKLOSW26}.

\section{Two-Slice Strength and Homogeneous Multiplicative Complexity}

\label{sect:mult}

This section presents short technical preliminaries and proves our main lemma (\cref{lem:strength})
connecting two-slice strength and homogeneous multiplicative complexity. We conclude by 
connecting $\lambda$-rank and $\lambda$-strength via a routine set-multilinearization 
(\cref{lem:set-multilinearization}).

\subsection*{Preliminaries---tensors as polynomials, \texorpdfstring{$\lambda$-rank}{lambda-rank}, and \texorpdfstring{$\lambda$-strength}{lambda-strength}}
\subsubsection*{Notation.} Let $\mathbb N=\{0,1,\ldots\}$ and let $\mathbb F$ be a field. 
We work with tensors over $\mathbb F$ in coordinates
and equip the vector space $\mathbb F^m$ with the standard basis $e_1,\ldots,e_m$. 
We write $\mathbb F[X_1,\ldots,X_q]$ for the $\mathbb F$-algebra of polynomials in the indeterminates $X_1,\ldots,X_q$. We write $\mathbb F[X_1,\ldots,X_q]_d$ for the linear subspace of homogeneous polynomials of degree $d\in\mathbb N$. A natural basis for these latter spaces consists of the $\binom{q+d-1}{d}$ monomials $X^a$ for $a \in \mathbb N^q$ and degree $d=\sum_{j=1}^q a_j$. 
We write $\deg f$ for the largest degree of a monomial appearing in $f$.
A related object is the set of tensors of order $d$ over $\mathbb F$, written $(\mathbb F^m)^{\otimes d}$. Its natural basis consists of the $m^d$ tensors $e_{i_1} \otimes \cdots \otimes e_{i_d}$ with $i_j \in [m]$ for $j \in [d].$

\subsubsection*{Tensors as Polynomials.} It is customary to understand polynomials as symmetric tensors via the identification 
\[
X_{i_1}\cdots X_{i_d} \mapsto \sum_{\sigma \in \mathfrak S_d} e_{i_{\sigma(1)}} \otimes \cdots \otimes e_{i_{\sigma(d)}},
\]
where $\mathfrak S_d$ is the symmetric group on $[d]$. 
We take a different (but also standard) approach, and understand tensors as specific 
multilinear polynomials.
Namely, we identify $(\mathbb F^m)^{\otimes d}$ with homogeneous degree-$d$ polynomials over $m\cdot d$ indeterminates $X^{(j)}_i$ for $j \in [d]$ and $i \in [m]$, as the image of $(\mathbb F^m)^{\otimes d}$ under the (injective) map sending 
\[
e_{i_1} \otimes \cdots \otimes e_{i_d} \mapsto X^{(1)}_{i_1} \cdots X^{(d)}_{i_d},
\]
and extended by linearity. 
Conversely, for $\emptyset\neq P\subseteq [d]$ we say that a monomial is $P$-{\em multilinear} if 
the monomial is of the form $\prod_{j\in P}X^{(j)}_{i_j}$ for some $i_j\in[m]$ for $j\in P$;
a polynomial is $P$-{\em multilinear} if all of its monomials are $P$-multilinear. 
The $d$-tensors in $(\mathbb F^m)^{\otimes d}$ are thus precisely 
the $[d]$-multilinear polynomials, which we also call {\em set-multilinear} polynomials.

\subsubsection*{$\lambda$-Rank and $\lambda$-Strength.} A {\em set partition} of $[d]$ is a set $\{P_1,\ldots,P_k\}$ with $\emptyset\neq P_i\subseteq[d]$ and $P_i\cap P_j=\emptyset$ for all $1\leq i<j\leq k$ and $\cup_{i=1}^k P_i=[d]$. 
For a nonempty set $\mathcal{P}$ of set partitions of $[d]$, we say that a
tensor $g\in(\mathbb{F}^m)^{\otimes d}$ admits a {\em $\mathcal{P}$-factorization} if 
there exists a set partition $\{P_1,\ldots,P_k\}\in\mathcal{P}$ with
$g=g_1\cdots g_k$ such that $g_i$ is $P_i$-multilinear for all $1\leq i\leq k$.
The \emph{$\mathcal{P}$-rank} of a tensor $f\in(\mathbb{F}^m)^{\otimes d}$ is
the minimum $r\in\mathbb{N}$ such that $f = f_1 + \cdots + f_r$ and each $f_i$ admits a
$\mathcal{P}$-factorization for all $1\leq i\leq r$.

For a positive integer $d$, we write $\lambda\vdash d$ to indicate that 
$\lambda=(\lambda_1,\ldots,\lambda_k)$ with
integers $\lambda_1\geq\cdots\geq\lambda_k\geq 1$ and $\sum_{i=1}^k\lambda_i=d$ is an integer partition of $d$ into $k$ parts 
for some positive integer $k$. We write $\lambda!=(\lambda_1!)\cdots(\lambda_k!)$.
For a tensor $f\in (\mathbb F^m)^{\otimes d}$, 
the \emph{$\lambda$-rank} $R_\lambda(f)$ of $f$ is the $\mathcal{P}$-rank of $f$, 
where $\mathcal{P}$ consists of all set partitions $\{P_1,\ldots,P_k\}$ of $[d]$ with 
$|P_i|=\lambda_i$ for all $1\leq i\leq k$. 
For $\lambda\vdash d$ we say that a polynomial $g$ admits a {\em $\lambda$-factorization} if $g=g_1\cdots g_k$ such that $g_i$ is homogeneous of degree $\lambda_i$ 
for all $1\leq i\leq k$.
The \emph{$\lambda$-strength} $S_\lambda(f)$ of a homogeneous polynomial $f$ of degree $d$ 
is the minimum $s\in\mathbb{N}$ such that $f = f_1 + \cdots + f_s$ and each $f_i$ admits a
$\lambda$-factorization for all $1\leq i\leq s$.

\begin{remark}
\label{rem:refinement}
For two partitions $\lambda,\mu\vdash d$ in {\em refinement order} $\lambda\sqsupseteq \mu$, such that $\lambda$ can be obtained from $\mu$ by replacing zero or more parts with their sum, 
it is immediate that both $R_\lambda\leq R_\mu$ and $S_\lambda\leq S_\mu$.
For $d=4$ in particular, we observe the two refinement chains 
\[
4\quad \begin{matrix}\rotatebox[origin=c]{30}{$\sqsupseteq$} \\[6pt] \rotatebox[origin=c]{-30}{$\sqsupseteq$} \end{matrix}\quad \begin{matrix} 3+1 \\[4.5pt]  ~ \\[4.5pt]  2+2\end{matrix}\quad \begin{matrix}\rotatebox[origin=c]{-30}{$\sqsupseteq$} \\[6pt] \rotatebox[origin=c]{30}{$\sqsupseteq$} \end{matrix}\quad 2+1+1 \quad \sqsupseteq \quad 1+1+1+1\,,
\]
which together with \eqref{eq:pip} imply
that the tightest lower-bound control on the multiplicative complexity 
is given via $R_{2,1,1}$ by \Cref{thm:main}, and on the homogeneous multiplicative
complexity via $S_{2,1,1}$ by \Cref{lem:strength}.
Given that any advancement on unconditional lower bounds for, e.g., the hyperclique $4$-tensor \eqref{eq:hyp}, would constitute major progress,
this already motivates our focus on two-slice rank and strength.
In addition, also for degrees $d\geq 5$, the two-slice rank $R_{d-2,1,1}$ and strength $S_{d-2,1,1}$ appear as the simplest non-trivial candidate measure arising from partition ranks that can also provide general lower bounds on multiplicative complexity.
\end{remark}

\subsection*{Two-Slice Decomposition of Homogeneous Multiplication Sequences}

We now prove our main lemma relating $(d-2,1,1)$-strength and 
homogeneous multiplicative complexity; we restate the lemma below for convenience.

\strlemma*

We will in fact prove the following slightly more general decomposition lemma,
and observe that \Cref{lem:strength} follows from \Cref{lem:seqlin} by taking $\Phi$ to be a homogeneous multiplication sequence of minimal length $\ell=H(f)$ computing $f$.

\begin{lemma}[Two-slice decomposition of homogeneous multiplication sequences]\label{lem:seqlin}
For every homogeneous multiplication sequence $\Phi = (f_1,\ldots,f_\ell)$, there exist linear forms $\alpha_1,\beta_1,\ldots,\alpha_k,\beta_k$ with $k \leq \ell$ such that every homogeneous 
polynomial $h$ computed by $\Phi$ with $\deg(h) \ge 2$ admits a decomposition
\[
h=\sum_{i=1}^k \alpha_i\beta_i g_i^{[h]}\,,
\]
where each polynomial $g_i^{[h]}$ is homogeneous and satisfies $\deg(g_i^{[h]}) \leq \deg(h)-2$.
\end{lemma}

\begin{proof}
We proceed by induction on $\ell$. If $\ell = 0$, then $\Phi$ by definition only computes constants and linear forms--both of degree less than $2$, so there is nothing to prove.
So suppose that $\ell \geq 1$, write $\hat \Phi=(f_1,\ldots,f_{\ell-1})$, and let $\hat\alpha_1,\hat\beta_1,\ldots,\hat\alpha_{k},\hat\beta_{k}$ be 
the linear forms afforded by the inductive hypothesis applied to $\hat \Phi$ with $k \leq \ell-1$.
Now study the polynomial $f_\ell$ in $\Phi$. We have that $f_\ell = st$ for some homogeneous
polynomials $s,t$ computed by $\hat \Phi$. Furthermore, we may assume $1 \leq \deg(s) \leq \deg(t)$.
We now split into two cases. In the first case, suppose that $\deg(s)=\deg(t)=1$.
We define
\[
\alpha_j=\hat\alpha_j\,,\quad\beta_j=\hat\beta_j\quad \text{for $1 \leq j \leq k$\quad and}
\quad
\alpha_{k+1}=s\,,\quad \beta_{k+1}=t\,.
\]
Now let $h$ be an arbitrary homogeneous polynomial of degree $d \geq 2$ computed by $\Phi$. 
Without loss of generality we have $h=u+\sigma f_\ell$ for some homogeneous polynomial
$u$ computed by $\hat \Phi$ and some $\sigma\in\mathbb F$. 
Now, since $u$ is computed by $\hat \Phi$, we have 
$u=\sum_{j=1}^{k} \hat\alpha_j\hat\beta_j\hat g_j^{[u]}$ 
by the inductive hypothesis.
Choose $g_j^{[h]} = \hat g_j^{[u]}$ for all $1 \leq j \leq k$
as well as $g_{k+1}^{[h]} = \sigma$. 
These choices obviously satisfy $\deg(g_j^{[h]}) \leq \deg(h)-2$ for all $1 \leq j \leq k+1$, 
and moreover
\[
h = u + \sigma f_\ell = \sum_{j=1}^{k} \hat\alpha_j\hat\beta_j \hat g_j^{[u]} + \alpha_{k+1} \beta_{k+1} g_{k+1}^{[h]} = \sum_{j=1}^{k+1} \alpha_j \beta_j g_j^{[h]}\,.
\]
This completes the inductive step in the first case. In the second case, 
suppose that $\deg(t) \geq 2$.
We define
\[
\alpha_j=\hat\alpha_j\,,\quad\beta_j=\hat\beta_j\quad \text{for $1 \leq j \leq k$}\,.
\]
Now let $h$ be an arbitrary homogeneous polynomial of degree $d \geq 2$ computed by $\Phi$. 
Without loss of generality we have $h=u+\sigma f_\ell=u+\sigma st$ for some homogeneous polynomial
$u$ computed by $\hat \Phi$ and some $\sigma\in\mathbb F$.
Now choose $g_j^{[h]} = \hat g_j^{[u]} + \sigma s \hat g_j^{[t]}$ for $1 \leq j \leq k$. 
By the inductive hypothesis applied to $u$ and $t$, it follows that
\[
h 
= u + \sigma st
=\sum_{j=1}^{k} \hat\alpha_j\hat\beta_j\hat g_j^{[u]} + \sigma s\sum_{j=1}^{k} \hat\alpha_j\hat\beta_j\hat g_j^{[t]}
=\sum_{j=1}^{k} \hat\alpha_j\hat\beta_j\bigl(\hat g_j^{[u]} + \sigma s\hat g_j^{[t]}\bigr)
=\sum_{j=1}^{k} \alpha_j\beta_jg_j^{[h]}\,.
\]
Moreover, since $u$ and $t$ are homogeneous, we may replace each $\hat g_j^{[u]}$ by its homogeneous component of degree $d-2$, and each $\hat g_j^{[t]}$ by its homogeneous component of degree $\deg(t)-2$, without changing the decompositions of $u$ and $t$. Hence we may assume that $\hat g_j^{[u]}$ 
and $\sigma s \hat g_j^{[t]}$ are both homogeneous of degree $d-2$, and therefore $g_j^{[h]}$ 
is homogeneous as well. This completes the inductive step in the second case and thus the proof.
\end{proof}

\begin{remark} \label{rem:random}
    Note that in a two-slice strength decomposition of elements of $\mathbb F[X_1,\ldots,X_q]_d$ into $s$ parts, we have $s \cdot D(q,d)$ degrees of freedom, where $D(q,d) = 2q + \binom{d+q-3}{d-2}$.
    In contrast, elements of $\mathbb F[X_1,\ldots,X_q]_d$ have $M(q,d)  = \binom{d+q-1}{d}$ degrees of freedom.
    Hence we heuristically expect $s \geq M(q,d)/D(q,d)$ to hold for a generic choice 
    $f\in \mathbb F[X_1,\ldots,X_q]_d$ (where we assume $\mathbb F$ is infinite).
    For a constant $d\geq 3$ we have $D(q,d) = \Omega(q^{d-2})$ and $M(q,d) = O(q^d)$, so that
    heuristically $s = \Omega(q^2)$ follows. For constant $d\geq 3$ we can similarly heuristically conclude $R_{d-2,1,1}(f)=\Omega(m^2)$ 
    for a generic choice $f\in(\mathbb F^m)^{\otimes d}$.
    These arguments can likely be made rigorous, but this is beyond the scope of the present work.
\end{remark}

\subsection*{Set-Multilinearization}

The following lemma relates $\lambda$-strength and $\lambda$-rank of $d$-tensors, proving \cref{eq:set-multilinearization}; we supply a proof for completeness.

\begin{lemma}[Set-multilinearization of $d$-tensors]
\label{lem:set-multilinearization}
For every $d$-tensor $f \in (\mathbb F^m)^{\otimes d}$ and every $\lambda \vdash d$, we have
\[
S_\lambda(f)\leq R_\lambda(f)\leq \frac{d!}{\lambda!}\cdot S_\lambda(f)\,.
\]
\end{lemma}

\begin{proof}
    The lower bound $S_\lambda(f)\leq R_\lambda(f)$ is immediate.
    By sub-additivity of $R_\lambda$, without loss of generality we may assume that $S_{\lambda}(f) = 1$, and show that $R_{\lambda}(f) \leq d!/\lambda!$.
    Suppose that $\lambda$ has $k$ parts, and assume $f$ admits the $\lambda$-factorization 
    $f=g_1\cdots g_k$. Since $f$ is a tensor and thus set-multilinear, without loss of generality 
    we may assume that all
    indeterminates of all monomials of each $g_i$ have degree at most one for $1\leq i\leq k$.
    For a subset $P \subseteq [d]$ of size $\lambda_i$, we write $g_i|_P$ for the polynomial $g_i$ restricted to only monomials of the form $\prod_{j \in P} X^{(j)}_{i_j}$ for some choice of $i_j\in[m]$ for all $j \in P$. (Equivalently, set $X^{(j)}_\ell = 0$ for all $\ell\in[m]$ and
    $j \in[d]\setminus P$ in $g_i$ to obtain $g_i|_P$.) Since $g_i$ is homogeneous of degree $\lambda_i$, we have 
    \[
    g_i = \sum_{\substack{P \subseteq [d] \\ |P| = \lambda_i}} g_i|_P\,.
    \]
    It follows that
    \[
    g_1 \cdots g_k = \prod_{i=1}^k \sum_{\substack{P \subseteq [d] \\ |P| = \lambda_i}} g_i|_P = \sum_{\substack{P_1,\ldots,P_k \subseteq [d] \\ \forall i\colon |P_i| = \lambda_i}} \prod_{i=1}^k g_i|_{P_i} = \sum_{\substack{P_1,\ldots,P_k \subseteq [d] \\ \forall i\colon |P_i| = \lambda_i \\ \forall i \neq j\colon P_i \cap P_j = \emptyset}} \prod_{i=1}^k g_i|_{P_i}\,,
    \]
    where the last equality stems from the fact that $f$ is by assumption a tensor and hence set-multilinear, so that all terms in the summation that do not come from a set partition of $[d]$ have to cancel eventually. 
    Now, per definition, $R_{\lambda}(\prod_{i=1}^k g_i|_{P_i}) = 1$ whenever the sets $P_i$ partition $[d]$ with $|P_i| = \lambda_i$, and the bound follows from the fact that the final rightmost sum can comprise at most $d!/\lambda!$ terms.
\end{proof}

\section{NP-Hardness of Slice Strength and Two-Slice Strength}

\label{sect:hard}

Throughout this section we assume that $\mathbb F$ is a field of characteristic zero.
This section proves \cref{thm:np-hardness-of-slice-strength-and-two-slice-strength}
in two parts, namely \cref{lem:np-hard-slice} and \cref{lem:np-hard-two-slice} below.

It will be convenient to work with monomials over the indeterminates $X_1,\ldots,X_q$ 
for $q\geq 1$ using the corresponding degree sequences
in $\mathbb N^q$ and a partial order $\preceq$ for such sequences. 
We start with short combinatorial preliminaries towards this end. 
A {\em composition} of $d\in\mathbb N$ into $q$ parts is 
a $q$-tuple $a=(a_1,\ldots,a_q)\in\mathbb N^q$ 
with $|a|=d=\sum_{j\in[q]}a_j$. Each such composition corresponds to
a unique monomial $X^a=\prod_{j\in [q]}X_j^{a_j}$ with {\em degree sequence}
$a$ and vice versa. We write $a!=(a_1!)\cdots(a_q!)$.
For two compositions $a,b\in\mathbb N^q$ we write $a\wedge b$ for the
elementwise minimum of the two compositions. 
We tacitly identify each subset
$S\subseteq[q]$ with its indicator-composition 
$S=(\mathbf 1_{j\in S}:j\in [q])\in\{0,1\}^q$
and write $X^S=\prod_{j\in S}X_j$ for the corresponding multilinear monomial.
We say that $a,b\in\mathbb N^q$
are in {\em Gale order} and write $a\preceq b$ if $\sum_{i=1}^j a_j\geq\sum_{i=1}^j b_j$ holds for all $j\in [q]$. Gale order is additive; namely, for $a,a',b,b'\in\mathbb N^q$ it holds that $a\preceq b$ and $a'\preceq b'$ imply $a+a'\preceq b+b'$.
We say that a set $H\subseteq[q]$ is a {\em hitting set} of $A\subseteq\mathbb N^q$ if $a\wedge H\neq (0,\ldots,0)$ holds for all $a\in A$. 
We write $\mathrm{hs}(A)$ for the minimum $k\in\mathbb N$ such that there
is a hitting set $H\subseteq[q]$ of $A$ with $|H|=k$.

The {\em support} $\vsupp f\subseteq\mathbb N^q$ of a polynomial $f\in\mathbb F[X_1,\ldots,X_q]$ consists of the degree sequences 
$a\in\mathbb{N}^q$ of all the 
monomials $X^a$ of $f$ with a nonzero coefficient.

\begin{lemma}[Hitting-set upper bound for slice strength] \label{lem:hs-ub}
    For all $f\in \mathbb F[X_1,\ldots,X_q]_d$ it holds that $S_{d-1,1}(f) \leq \mathrm{hs}(\vsupp(f))$.
\end{lemma}
\begin{proof}
    Let $H = \{i_1,\ldots,i_k\} \subseteq [q]$ be a hitting set of $\vsupp(f)$. 
    Partition the set $\vsupp(f) = A_1\cup \cdots \cup A_k$ arbitrarily subject to $a_{i_j}\geq 1$ whenever $a \in A_j$.
    Write $f_j$ for the sum over all monomials 
    $c_a \cdot X^a$ (including their coefficients $c_a \neq 0$) of $f$ such that $a \in A_j$.
    Since the $A_j$ partition $\vsupp(f)$, it follows that $f = \sum_{j=1}^k f_j$.
    By definition of $A_j$, $X_{i_j}$ divides every monomial in $f_j$
    and hence $f_j$ itself, so that $f_j = X_{i_j} \cdot g_j$ holds for some $g_j$. Therefore, $f$ admits a representation as 
    $f = \sum_{j=1}^k X_{i_j} \cdot g_j$, proving $S_{d-1,1}(f) \leq k$.
\end{proof}

The following lower bound for $S_{d-1,1}$
is structured for analogy with the original
proof of Sawin and Tao for $R_{d-1,1}$. Accordingly, 
endow the space  
$\mathbb F[X_1,\ldots,X_q]_d$ with the symmetric bilinear form defined 
for all $a,b\in\mathbb N^q$ with $|a|=|b|=d$ by 
extension from the rule
\begin{equation}
\label{eq:ip}
\langle X^a, X^b \rangle = \begin{cases} a!/d!\, &\text{if $a=b$}\,,\\ 
0\, &\text{otherwise}.\end{cases}
\end{equation}
In particular, the standard basis of monomials $X^a$ for $a\in\mathbb{N}^q$ with $|a|=d$ is orthogonal under this bilinear form.
For a subset $A\subseteq\mathbb N^q$, let us write $\max A\subseteq A$ 
for the subset of all $\preceq$-maximal elements of $A$. 

\begin{lemma}[Hitting-set lower bound for slice strength] \label{lem:hs-lb}
    For all $f\in \mathbb F[X_1,\ldots,X_q]_d$ 
    it holds that $S_{d-1,1}(f) \geq \mathrm{hs}(\max\vsupp(f))$.
\end{lemma}
\begin{proof}
    Suppose that $f = \sum_{i=1}^s \ell_i g_i$ for linear forms $\ell_i \in \mathbb F[X_1,\ldots,X_q]_1$, where $s = S_{d-1,1}(f)$.
    Then, the vector space 
    \[
    \hat U = \{ \hat u \in \mathbb F^q : \ell_i(\hat u) = 0\ \text{for all $i\in[s]$}\}
    \]
    has linear dimension $k \geq q-s$.
    By definition, $f(\hat u) = \sum_{i=1}^s \ell_i(\hat u) g_i(\hat u) = 0$ holds for all $\hat u \in \hat U$.
    Identify each vector $\hat u\in \hat U$ with the linear form
    $u=\sum_{j\in[q]}\hat u_jX_j\in\mathbb{F}[X_1,\ldots,X_q]_1$, and observe that the set of all
    such linear forms defines a vector space 
    $U\subseteq\mathbb{F}[X_1,\ldots,X_q]_1$.
    We now have 
    $\langle f,u^d \rangle = 0$ for all $u \in U$; indeed, observe that for all $a \in \mathbb N^q$ with $|a| = d$, we have the identity
    $\langle X^a, u^d\rangle=X^a(\hat u)$ by \eqref{eq:ip} and
    thus by linearity $\langle f,u^d\rangle = f(\hat u)=0$
    for all $u\in U$. Using linearity and the polarization identity 
    \[
    u_1 \cdots u_d = \frac{1}{2^d d!} \sum_{\sigma \in \{-1,1\}^{d}} \left( \prod_{j=1}^d \sigma(j)\right) \cdot \left( \sum_{i=1}^d \sigma(i)u_i  \right)^d
    \]
    it follows that $\langle f,u_1 \cdots u_d\rangle = 0$ holds for every choice of $u_1,\ldots,u_d \in U$ as well.
    
    Let $u_1,\ldots,u_{k} \in U$ be the unique basis of $U$ in reduced row-echelon form with respect to the basis $X_1,\ldots,X_q$ for the containing vector space $\mathbb{F}[X_1,\ldots,X_q]_1$. 
    That is, there is a sequence $1 \leq i_1 < \cdots < i_k \leq q$ such that $u_j = X_{i_j} + u_j'$, where the monomials of $u_j'$ contain only indeterminates $X_\ell$ with indices $\ell > i_j$.
    Now let $m=(m_1,\ldots,m_k)\in\mathbb{N}^k$ with $m_1+\ldots+m_k=d$
    and write $a^{[m]}\in \mathbb N^q$ for the composition that
    satisfies $|a^{[m]}|=d$ and $a_{i_j}^{[m]}=m_j$ for all $j\in[k]$.
    By the structure of the basis $u_1,\ldots,u_k$ and additivity of Gale order $\preceq$, 
    the product $u^m=u_1^{m_1}\cdots u_k^{m_k}$ has the property that
    $a^{[m]}$ is the $\preceq$-minimum element of $\vsupp(u^m)$.
    Furthermore, $\langle f,u^m\rangle=0$. Now observe that for all $b\in\vsupp(f)$
    it holds that $\langle f,X^b\rangle\neq 0$. We claim that
    $a^{[m]}\notin\max\vsupp(f)$. To reach a contradiction, suppose that
    $a^{[m]}\in\max\vsupp(f)$. Then $X^{a^{[m]}}$ contributes a 
    nonzero term to $\langle f,u^m\rangle$, implying by $\langle f,u^m\rangle=0$ that there 
    exists a $b\in\vsupp(u^m)\setminus\{a^{[m]}\}$ such that 
    $X^b$ also contributes a 
    nonzero term to $\langle f,u^m\rangle$. Since $b\in\vsupp(f)$ and $a^{[m]}\prec b$, we obtain a contradiction to $a^{[m]}\in\max\vsupp(f)$.
    Since $m$ was arbitrary, it follows that every $a\in\max\vsupp(f)$
    has the property that there exists an $\ell\in H:=[q]\setminus\{i_1,i_2,\ldots,i_k\}$ with $a_\ell\geq 1$. That is, $H$ is a 
    hitting set of $\max\vsupp(f)$ with $|H|=q-k\leq s$.
\end{proof}

\begin{lemma}[NP-hardness of slice strength for degree $d\geq 3$] \label{lem:np-hard-slice}
    Given as input a list of monomials for a degree $d\geq 3$ homogeneous multilinear polynomial $f$ with coefficients in $\{0,1\}$,
    it is $\mathrm{NP}$-hard to compute $S_{d-1,1}(f)$.
\end{lemma}
\begin{proof}
    From \Cref{lem:hs-lb} and \Cref{lem:hs-ub} we have $S_{d-1,1}(f) = \mathrm{hs}(\vsupp(f))$ whenever $\vsupp(f)$ is an antichain with respect to the Gale order $\preceq$. 
    Thus, it suffices to show that the hitting set problem remains $\mathrm{NP}$-hard over an ordered universe with inputs consisting of such antichains.
    Given as input an instance $S_1,\ldots,S_m\subseteq [n]$ of hitting set with 
    $|S_1|=\cdots=|S_m|=d-1$. Take $q=n+m$ and set $S_i'=S_i\cup\{q+1-i\}\subseteq[q]$ 
    for $i\in[m]$. 
    Without loss of generality we may assume that $S_1,\ldots,S_m$ are listed in a
    linearization of the Gale order; namely, $S_i\preceq S_j$ implies $i\leq j$ for all $i,j\in[m]$. 
    (This can be easily ensured in polynomial time via a topological sorting on the $S_i$.)
    Now, every hitting set of size $k$ for $S_1',\ldots,S_m'\subseteq[q]$ 
    corresponds to a hitting set of size $k$ of $S_1,\ldots,S_m\subseteq[n]$. Indeed, every occurrence of $q+1-i$ for $i\in[m]$ in a hitting set is necessarily superfluous, as it can always be replaced by an arbitrary element of $[n]$ contained already in $S_i$. This yields a hitting set of size at most $k$ of the original instance. Conversely, every hitting set of $S_1,\ldots,S_m$ is already a hitting set of $S_1',\ldots,S_m'$, 
    so there is nothing to show.
    Moreover, let $i,j \in [m]$ with $i \neq j$. If $S_i$ and $S_j$ were already Gale-incomparable in the original instance, then so are $S'_i$ and $S'_j$.
    If $S_i \prec S_j$, then $i < j$ and so $q+1-i > q+1-j$ by definition, hence $S'_i$ and $S'_j$ are again Gale-incomparable.
    
    Now, given an input instance for hitting set, it is enough to write down the 
    polynomial $f = \sum_{i=1}^m X^{S_i'}\in\mathbb F[X_1,\ldots,X_q]$, which clearly satisfies the restrictions posed in the statement of the claim. We also note that the hitting set problem remains $\mathrm{NP}$-hard when $|S_i|=d-1 \geq 2$; indeed, for $d-1=2$ the hitting set problem is the minimum vertex cover problem, and hence $|S'_i| = d\geq 3$ holds for all $i\in[m]$.
\end{proof}

\begin{remark}
    It should be noted that our result on slice strength does not imply the result of Bläser et al.~\cite{BlaserILPS21} on slice rank, because the correspondence between slice rank and slice strength is not straightforward. Moreover, they are interested in other specific properties of slice rank that may make it useful for building intuition in the context of geometric complexity theory.
    We do not know whether slice strength also satisfies these properties. 
\end{remark}

\begin{lemma}[NP-hardness of two-slice strength for degree $d\geq 4$]
\label{lem:np-hard-two-slice}
   Given as input a list of monomials for a degree $d\geq 4$ homogeneous multilinear polynomial $f$ with coefficients in $\{0,1\}$,
    it is $\mathrm{NP}$-hard to compute $S_{d-2,1,1}(f)$.
\end{lemma}
\begin{proof}
    As in the proof of~\Cref{lem:np-hard-slice}, consider an instance of the vertex cover (that is, $2$-uniform hitting set) problem on a simple loopless graph $G$ with vertex set
    $V(G)=[n]$ and edge set $E(G)$, where each edge is a 2-element subset of $V(G)$. 
    Let $\mathrm{vc}(G)$ be the size of a minimum vertex cover on $G$.
    Now, take $q=2n$ and
    construct the multilinear polynomial 
    $f = \sum_{\{u,v\} \in E(G)} X_u X_v X_{q+1-u} X_{q+1-v}\in\mathbb F[X_1,\ldots,X_q]$.
    Since every vertex cover of $G$ corresponds to a hitting set of the same size of the degree sequences of the monomials of $f$ (and vice-versa for minimum hitting sets; this is analogous to the proof of~\cref{lem:np-hard-slice}), it follows that 
    \begin{align} \label{eq:vc-hs}
    \mathrm{vc}(G) = \mathrm{hs}(\vsupp(f))\,.
    \end{align}
    Moreover, by the structure of $f$, vertex cover $C$ corresponds to a decomposition $f = \sum_{c \in C} X_c X_{q+1-c} g_c$, in analogy with the proof of~\Cref{lem:hs-ub} (although for $S_{d-2,1,1}$, this upper bound hinges on the specific construction for $f$): We can arbitrarily partition the monomials of $f$ in a way compatible with $C$, and then factor out $X_c X_{q+1-c}$. 
    This shows
    \begin{align} \label{eq:str-hs}
    S_{d-2,1,1}(f) \leq \mathrm{hs}(\vsupp (f))
    \end{align}
    for this particular choice of $f$.
    
    In addition, by the choice of variable ordering, the degree sequences of monomials of $f$ are already a Gale-antichain: If $\{u,v\} \prec \{u', v'\}$ (and hence $u\leq u'$ and $v \leq v'$), then the augmented degree sequences corresponding to $\{u,v,q+1-u,q+1-v\}$ and $\{u',v',q+1-u',q+1-v'\}$ are incomparable. If $\{u,v\}$ and $\{u',v'\}$ were already incomparable, then they remain so after augmentation. It follows that 
    \begin{align} \label{eq:vsupp}
    \vsupp(f) = \max \vsupp(f)\,.
    \end{align}
    Thus \Cref{lem:hs-lb,lem:hs-ub} in combination with~\Cref{eq:vc-hs,eq:str-hs,eq:vsupp} and the fact that $S_{d-1,1}(f) \leq S_{d-2,1,1}(f)$ together imply 
    \[
    S_{d-2,1,1}(f) \leq \mathrm{vc}(G) = \mathrm{hs}(\vsupp(f)) = S_{d-1,1}(f) \leq S_{d-2,1,1}(f)
    \]
    so that $S_{d-2,1,1}(f) = \mathrm{vc}(G)$ as desired, proving $\mathrm{NP}$-hardness.
    The argument can be extended to $d \geq 5$ via padding the monomials with fresh variables per edge, which does not lower $S_{d-2,1,1}$.
\end{proof}

\begin{remark}
    The polynomials produced by the preceding reduction all have two-slice rank bounded by $O(n)$, so that they cannot be used to construct explicit examples of polynomials with superlinear two-slice rank.
\end{remark}

\section*{Acknowledgment}
We would like to thank Markus Bläser and Radu Curticapean for valuable discussions, and the anonymous referees for their valuable comments.

\bibliographystyle{plainurl}
\bibliography{notes.bib}

\end{document}